\documentclass[11pt,a4paper]{article}
\pdfoutput=1

\usepackage[utf8]{inputenc}
\usepackage[bookmarks]{hyperref}
\hypersetup{pdftitle = {Probing multipartite entanglement through persistent homology}, pdfauthor = {Gregory A. Hamilton, Felix Leditzky},colorlinks=true,citecolor=blue,linkcolor=blue,filecolor=blue,urlcolor=blue}
\usepackage{amsmath,amssymb,amsthm}
\usepackage{dsfont}
\usepackage{fullpage}

\usepackage[affil-it]{authblk}

\usepackage{cleveref}
\crefformat{enumi}{(#2#1#3)}
\crefmultiformat{enumi}{(#2#1#3)}{ and~(#2#1#3)}{, (#2#1#3)}{ and~(#2#1#3)}
\crefrangeformat{enumi}{(#3#1#4)--(#5#2#6)}

\usepackage{enumerate}
\usepackage{mathtools,bm}

\usepackage{booktabs}

\allowdisplaybreaks[4]

\numberwithin{equation}{section}

\newtheorem{theorem}{Theorem}
\newtheorem{lemma}[theorem]{Lemma}

\newtheorem{corollary}[theorem]{Corollary}
\theoremstyle{definition}

\usepackage{tikz}
\usepackage{pgfplots}
\usepackage{xcolor}
\usepackage[textsize=footnotesize]{todonotes}

\usetikzlibrary{cd}

\DeclareMathOperator{\Tr}{Tr}
\DeclareMathOperator{\rk}{rk}
\newcommand*{\hilb}{\mathcal{H}}
\newcommand*{\hal}{\mathcal{A}}
\newcommand{\cC}{\mathcal{C}}
\newcommand{\cD}{\mathcal{D}}
\newcommand{\cH}{\mathcal{H}}
\newcommand{\cA}{\mathcal{A}}
\newcommand{\eps}{\varepsilon}
\newcommand{\cF}{\mathcal{F}}
\newcommand{\bD}{\mathbf{D}}
\newcommand{\bE}{\mathbf{E}}

\usepackage[backend=bibtex8,sorting=nty,firstinits=true,doi=false,isbn=false,url=false,maxbibnames=6]{biblatex}
\renewbibmacro{in:}{}
\newbibmacro{string+doi}[1]{\iffieldundef{doi}{#1}{\href{https://dx.doi.org/\thefield{doi}}{#1}}}
\DeclareFieldFormat{title}{\usebibmacro{string+doi}{\mkbibemph{#1}}}
\DeclareFieldFormat[article]{title}{\usebibmacro{string+doi}{\mkbibquote{#1}}}
\DeclareFieldFormat[incollection]{title}{\usebibmacro{string+doi}{\mkbibquote{#1}}}                   
\DeclareFieldFormat[inproceedings]{title}{\usebibmacro{string+doi}{\mkbibquote{#1}}} 

\DeclareMathOperator{\supp}{supp}

\usepackage[normalem]{ulem}

\begin{document}
	\title{Probing multipartite entanglement through persistent homology}
	
	\author[1,2,4]{Gregory A. Hamilton\thanks{\texttt{gah4@illinois.edu}}}

	\author[3,4]{Felix Leditzky\thanks{\texttt{leditzky@illinois.edu}}}
	\affil[1]{\normalsize Department of Physics, University of Illinois Urbana-Champaign}
	\affil[2]{Institute for Condensed Matter Theory, University of Illinois Urbana-Champaign}
	\affil[3]{Department of Mathematics, University of Illinois Urbana-Champaign}
	\affil[4]{Illinois Quantum Information Science and Technology Center (IQUIST),\protect\\ University of Illinois Urbana-Champaign}

	\date{\today}

\maketitle

\begin{abstract}
We propose a study of multipartite entanglement through persistent homology, a tool used in topological data analysis. In persistent homology, a 1-parameter filtration of simplicial complexes called persistence complex is used to reveal persistent topological features of the underlying data set.
This is achieved via the computation of homological invariants that can be visualized as a persistence barcode encoding all relevant topological information. In this work, we apply this technique to study multipartite quantum systems by interpreting the individual systems as vertices of a simplicial complex. To construct a persistence complex from a given multipartite quantum state, we use a generalization of the bipartite mutual information called the deformed total correlation. Computing the persistence barcodes of this complex yields a visualization or `topological fingerprint' of the multipartite entanglement in the quantum state. The barcodes can also be used to compute a topological summary called the integrated Euler characteristic of a persistence complex. We show that in our case this integrated Euler characteristic is equal to the deformed interaction information, another multipartite version of mutual information. When choosing the linear entropy as the underlying entropy, this deformed interaction information coincides with the $n$-tangle, a well-known entanglement measure. The persistence barcodes thus provide more fine-grained information about the entanglement structure than its topological summary, the $n$-tangle, alone, which we illustrate with examples of pairs of states with identical $n$-tangle but different barcodes.
Furthermore, a variant of persistent homology computed relative to a fixed subset yields an interesting connection to strong subadditivity and entropy inequalities. We also comment on a possible generalization of our approach to arbitrary resource theories.
\end{abstract}

\section{Introduction}
Entanglement is a strong form of correlation between quantum systems that lies at the heart of quantum information processing, giving rise to novel protocols in quantum algorithms \cite{shor1999polynomial,Jozsa2003}, quantum communication \cite{bennett1992communication,bennett1993teleporting,smith2008quantum}, quantum computation \cite{Gottesman1999,Knill2001,Nielsen2003,leung2004quantum,raussendorf2001one,raussendorf2002computational,raussendorf2003measurement}, quantum cryptography \cite{ekert1991crypto}, and quantum-enhanced sensing \cite{caves1981quantum,holland1993interferometric,wineland1994squeezed,bollinger1996measurements,Giovannetti2004}.
Due to the fundamental role of entanglement as a resource in these tasks, finding a succinct mathematical characterization of entangled quantum states is a central goal in quantum information theory.
In the simplest case of pure bipartite quantum states, the singular value decomposition gives an efficient tool to accomplish this goal.
Unfortunately, the more general cases of mixed quantum states or multipartite states on systems consisting of more than two parts are significantly more challenging.
This is partly due to the fact that the dimension of the state space grows exponentially with the number of constituents, and hence an exponentially large number of parameters is needed to describe states.
As a result, our understanding of entanglement even for pure multipartite quantum states is rudimentary at best. 

A common approach to studying multipartite entanglement is using functionals that measure the correlations in entangled states.
The desired properties of such functionals depend on the operational context.
For example, we may consider the setting where the individual parties of a multipartite system can only perform local operations and classical communication (LOCC) between each other.
Since entangled states cannot be created from scratch using LOCC protocols alone, they form a valuable resource in this context \cite{chitambar2019quantum}.
\emph{Entanglement measures}, functionals that cannot increase on average under LOCC, aim to quantify the usefulness of such a resource \cite{Vidal2000,Horodecki_Horodecki_Horodecki_Horodecki_2009}.
An important example of an entanglement measure in this work is the $n$-tangle, which was originally defined for three qubits \cite{Dur_Vidal_Cirac_2000} and can be generalized to multi-qubit systems \cite{wong2001multiparticle}.

The LOCC setting is well-motivated from an operational point of view, but rather intractable in mathematical terms~\cite{chitambar2014everything}.
One may consider the relaxed setting of \emph{stochastic} LOCC (SLOCC) consisting of LOCC protocols that achieve a state transformation only with some positive success probability.
Since SLOCC maps can be identified with local operators, this gives rise to a more tractable mathematical treatment \cite{Dur_Vidal_Cirac_2000,verstraete2002four}.
Similar to LOCC, there are measures that quantify entanglement in the SLOCC context.
Particularly interesting are measures that are invariant under invertible SLOCC transformations corresponding to local invertible operators. These measures can be used to distinguish different SLOCC equivalence classes~\cite{Dur_Vidal_Cirac_2000,verstraete2002four,verstraete2003normal,walter2013entanglement,Gour_Wallach_2013}.

Another useful way to measure the correlations in quantum systems is via their usefulness or operational relevance in an information-theoretic task.
The guiding principle is to consider an information-processing task that makes use of the correlations in the given quantum state, and showing that the optimal rate of achieving this task is equal to an entropic measure.
For example, the \emph{quantum mutual information} of a bipartite state quantifies the amount of local noise that has to be applied to the state in order to decorrelate it \cite{groisman2005correlations}.
In contrast to the entanglement measures mentioned above, the mutual information measures the total amount of correlation between two systems, including quantum as well as classical correlations.
There are different generalizations of the mutual information to multipartite systems; particularly interesting for us is the so-called \emph{total correlation}~\cite{watanabe1960information,herbut2004mutual,groisman2005correlations,wilde2015multipartite}, which enjoys a similar operational interpretation as the mutual information, as the total amount of local noise needed to fully decorrelate a multipartite state \cite{groisman2005correlations}.
Another multipartite generalization of the mutual information is the \emph{interaction information}, with applications in the study of topologically ordered systems \cite{kitaev2006topological}, scrambling \cite{Hosur2016,schnaack2019tripartite}, and the AdS/CFT correspondence \cite{Hayden2016,nezami2020multipartite}.

\subsection{Overview of main results}\label{sec:overview-results}
In this work, we propose a novel approach to characterizing multipartite entanglement based on a tool from topological data analysis called persistent homology.
The key idea is to build up a 1-parameter filtration of abstract simplicial complexes called a persistence complex, defined by measuring the correlations between quantum systems in a multipartite state using a suitable functional.
We choose a $q$-deformed Tsallis version of the total correlation as this functional.
Persistent homology then amounts to keeping track of the changes in topology as a function of the filtration parameter, thus revealing persistent topological features that characterize the underlying entanglement structure, as explained in Fig.~\ref{fig:example}.
Our main visualization tool is the persistence barcode of a persistence complex, which can be interpreted as a ``topological fingerprint'' of multipartite entanglement.

We find that the integrated Euler characteristic of the filtration complex induced by the $q$-deformed total correlation equals the $q$-deformed interaction information (\Cref{thm:interaction-information}).
Specializing to $q=2$, which amounts to choosing the linear entropy as the underlying entropy, this integrated Euler characteristic equals the $n$-tangle, the entanglement monotone mentioned above (\Cref{thm:n-tangle}). 
This result answers a question raised by Eltschka and Siewert~\cite{Eltschka2018distributionof} about the potentially topological nature of an entanglement measure called \emph{distributed concurrence} that can be expressed as a formula reminiscent of an Euler characteristic.
This distributed concurrence coincides (up to an irrelevant factor) with the integrated Euler characteristic of the persistence complex derived from the $q$-deformed total correlation via \Cref{thm:n-tangle}, thereby establishing the desired topological interpretation.

Another consequence of \Cref{thm:n-tangle} is that persistent homology provides a more fine-grained analysis of the entanglement structure of a multipartite state than the $n$-tangle.
We illustrate this point by exhibiting examples of SLOCC-inequivalent pairs of states with identical $n$-tangle values but differing persistent barcodes in Figures~\ref{fig:comparison-graph-states} and \ref{fig:comparison-osterloh-siewert}.
These examples demonstrate that the topological fingerprint conveys more information about the multipartite entanglement than the $n$-tangle alone, and we conjecture that they may be used to rigorously detect SLOCC-inequivalence.

Using a modified technique that computes persistent homology relative to a given subset, we then exhibit an interesting link between persistent homology and entropy inequalities: the integrated Euler characteristic of the relative persistence complex of a tripartite system equals the negative conditional mutual information, and is thus non-positive because of strong subadditivity (\Cref{thm:non-positivity-IEC}).

Finally, we discuss how our approach can naturally be generalized to other multipartite correlation functionals, in particular generalized divergences satisfying the data processing inequality.
We argue that the persistent homology arising from this choice is a promising tool to study more general resource theories.

\subsection{Related work}

Persistent homology has previously been used as a tool to study entanglement in \cite{dipierro2018homological,mengoni2019persistent,olsthoorn2021persistent}.
In these works, the individual systems of a multipartite system are treated as data points, and a \emph{bipartite} correlation measure is used to assign distances between pairs of systems.
In \cite{dipierro2018homological,mengoni2019persistent} the concurrence is used as the bipartite measure, whereas \cite{olsthoorn2021persistent} uses an inverse mutual information quantity.
The chosen distance measure is used to build the Vietoris-Rips complex, in which a subset of $k+1$ points forms a $k$-simplex whenever the distance of any two data points in the subset is at most equal to some fixed parameter $\eps$.
Varying the parameter $\eps$ then yields a filtration of simplicial complexes, and persistent homology tools are used to analyze the entanglement of the quantum state.

Here, we use a \emph{multipartite} correlation measure to build the persistence complex, which is explained in more detail in Section~\ref{sec:ph}.
The advantage of this approach is that multipartite correlations in $k$-simplices are more directly encoded in the persistent homology.
This leads to a more refined analysis of the entanglement structure of a quantum state and yields a well-known entanglement measure as a topological summary.  

While probing multipartite entanglement via homology is not a new notion \cite{mainiero2019homological,Baudot_2019,Baudot_Bennequin_2015,Vigneaux_2019}, our methodology outlined here is a novel yet simple mechanism by which to categorify the interaction information and several other entropic quantities. 
This application of persistent homology in the context of multipartite entanglement is to our knowledge unexplored in the literature, and we expect future work to clarify its relation to previous works.

\subsection{Structure of this manuscript}
The remainder of this manuscript is organized as follows.
In Section~\ref{sec:preliminaries} we fix some notation and define the entropic quantities used in this work.
We then review persistent homology and topological summaries in Section~\ref{sec:ph}, giving both an intuitive explanation of the key ideas as well as a formal exposition.
In Section~\ref{sec:main-results} we present our main results:
We first show in Section~\ref{sec:functionals} that the $q$-deformed total correlation is a valid functional for the persistent homology pipeline.
In Section~\ref{sec:correlation-measures} we prove the advertised results on the topological summaries of the persistence complex in terms of the total correlation, and discuss how persistent homology conveys more fine-grained information about the entanglement structure of a quantum state.
In Section~\ref{sec:reduced} we show how strong subadditivity implies a non-positivity condition of the integrated Euler characteristic.
We conclude in Section~\ref{sec:conclusions} with a summary of our results and future directions of research.

\section{Preliminaries}\label{sec:preliminaries}
\subsection{Notation}\label{sec:notation}

Given a finite dimensional Hilbert space $\hilb$ we denote by $\mathcal{B}(\hilb)$ the algebra of linear operators acting on $\hilb$. We set $|\hilb_{A}|\coloneqq \dim \hilb_{A}$, where $A$ is a quantum system, and write $X_{A_{1}\ldots A_{n}}$ for operators in $\mathcal{B}(\hilb_{A_{1}\cdots A_{n} })$, where $\hilb_{A_{1}\cdots A_{n}}\coloneqq \hilb_{A_{1}}\otimes \cdots \otimes \hilb_{A_{n}}$. We equivalently write $X_{A_{1}\cdots A_{n}} \equiv X_{\hal}$ for set $\hal\coloneqq\{A_{1},\ldots,A_{n}\}$. Note the distinction between $|\hilb_{A}|$ and $|A|$, the cardinality of a subset $A\subseteq \hal$. We denote by $\mathcal{P}(\hilb_{A}) \coloneqq \{\rho \in \mathcal{B}(\hilb_{A})\colon \rho \ge 0\}$ the set of positive semidefinite operators on $\hilb_{A}$, and by $\mathcal{D}(\hilb_{A}) \coloneqq \{\rho \in \mathcal{P}(\hilb_{A})\colon \Tr(\rho_{A}) = 1\}$ the set of quantum states or density operators on $\hilb_{A}$.
Throughout the paper $\log$ denotes the natural logarithm.

\subsection{Entropic measures}
For a state $\rho_{A}\in\cD(\cH_A)$ the Tsallis entropy $S_q(\rho_A)\equiv S_{q}(A)_{\rho}$ for $q\in (0,1)\cup(1,\infty)$ is defined as 
\begin{align} 
	S_{q}(A)_{\rho} = \frac{1}{1-q}\left(\Tr \rho_{A}^{q} -1\right).\label{eq:tsallis-entropy}
\end{align}
The quantity $S_{2}(A)_{\rho} = 1-\Tr\rho_A^2$ is called the \textit{linear entropy}, while the von Neumann entropy $S(\rho) = -\Tr\rho\log\rho$ is recovered in the limit $S(A)_{\rho} \equiv \lim_{q\to 1}S_{q}(A)_{\rho}$ \cite{Tsallis1988}. 
We sometimes drop the $\rho$ subscript when the quantum state is clear. 

The quantum relative entropy $D(\rho\|\sigma)$ is defined as
\begin{align}
	D(\rho\|\sigma) = \begin{cases} \Tr(\rho(\log\rho - \log \sigma)),\ & \text{if $\supp\rho \subseteq \supp \sigma$,}\\
		\infty & \text{otherwise,}
		\end{cases}
	\label{eq:quantum-relative-entropy}
\end{align}
where $\supp X$ denotes the orthogonal complement of the kernel of an operator $X$.
The relative entropy is an example of a generalized divergence (see \Cref{sec:conclusions}) as it satisfies the data processing inequality $D(\rho_{AB}\|\sigma_{AB}) \geq D(\rho_A\|\sigma_A)$ for all $\rho_{AB},\sigma_{AB}$ \cite{lieb1973ssa}.

The mutual information of a bipartite state $\rho_{AB}$ is defined as 
\begin{align}
	I(A:B)_\rho = S(A)_\rho + S(B)_\rho - S(AB)_\rho = D(\rho_{AB}\|\rho_A\otimes \rho_B).
\end{align}
It has an operational interpretation as the amount of local noise the parties $A$ and $B$ have to apply in order to remove all correlations between each other \cite{groisman2005correlations}.

There are different generalizations of the mutual information to multipartite systems $\cA=\lbrace A_1,\dots,A_n\rbrace$.
The \emph{total correlation}~\cite{watanabe1960information,herbut2004mutual,groisman2005correlations,wilde2015multipartite} is defined as
\begin{align}
	C(\cA)_\rho \coloneqq \sum_{i=1}^n S(A_i)_\rho - S(\cA)_\rho = D\Big(\rho_\cA\|\bigotimes\nolimits_{i=1}^n\rho_{A_i}\Big).
	\label{eq:total-correlation}
\end{align} 
It enjoys an operational interpretation similar to the mutual information, equalling the total amount of local noise needed to fully decorrelate a multipartite state \cite{groisman2005correlations}.
Another multipartite generalization of the mutual information is the \emph{interaction information}, defined as
\begin{align}
	I(\cA)_\rho \coloneqq \sum_{J\subseteq \cA}(-1)^{|J|-1}S(J)_\rho.
	\label{eq:interaction-information}
\end{align}
Specializing to three parties, the interaction information gives the tripartite information 
\begin{align}
	I(\lbrace A,B,C\rbrace) = S(A) +S(B)+S(C)-S(AB)-S(AC)-S(BC)+S(ABC),
\end{align} 
with applications in the study of many-body physics  \cite{kitaev2006topological,Hosur2016,schnaack2019tripartite}, the AdS/CFT correspondence \cite{Hayden2016,nezami2020multipartite}, and secret sharing \cite{junge2020tripartite}.

\section{Persistent Homology}\label{sec:ph}
\subsection{Intuitive picture}

In our work we view the parties $A_i$ of a multipartite system $\cA = \lbrace A_1,\dots,A_n\rbrace$ as vertices of an abstract simplicial complex, that is, a family of sets closed under taking subsets.
The elements of this simplicial complex are determined by a real-valued functional $F\colon 2^\cA\to\mathbb{R}$ on subsets of the quantum systems, evaluated using a given quantum state on $\cA$.
For a fixed value of a ``filtration parameter'' $\eps\in\mathbb{R}$, we add the simplex $J$ to the complex if $F(J)\leq \eps$.
We require the functional to be monotonic with respect to taking subsets, $F(I)\leq F(J)$ for $I\subseteq J$, such that this indeed defines a valid simplicial complex $G(\eps)$.
This yields a filtration of simplicial complexes $\lbrace G(\eps)\colon \eps\in\mathbb{R}\rbrace$ called the ``persistence complex'', with $G(\eps)\subseteq G(\eps')$ if $\eps\leq \eps'$.

Persistent homology provides a means of keeping track of how the topology changes as a function of the filtration parameter $\eps$ \cite{edelsbrunner2002persistence}.
To this end, we consider the homology groups $H_k(\cdot)$ of the simplicial complexes, which intuitively count the number of $k$-dimensional holes in a topological space.
The inclusion of simplicial complexes $G(\eps)\subseteq G(\eps')$ for $\eps\leq \eps'$ induces homomorphisms $H_k(G(\eps))\to H_k(G(\eps'))$ between homology groups of given dimension $k$, and the $k$-th persistent homology group is defined to be the image of $H_k(G(\eps))$ under this homomorphism \cite{otter2017roadmap}.
A persistent barcode is then a unique collection of intervals represented by stacked lines, one stack for each dimension $k$ \cite{Bauer_Lesnick_2015}.
The width of a line in the barcode corresponds to the interval of the filtration parameter $\eps$ for which the corresponding homology group is non-trivial; informally speaking, it visualizes how long clusters or holes in the topological space persist.
The number of barcodes of dimension $k$ spanning filtration value $\eps$ corresponds to the $k$-th Betti number $\beta_{k}(\eps)$, the rank of the $k$-th homology group at $\eps$.
We refer to Fig.~\ref{fig:example} for a simple example of a persistence complex together with its barcode that has been derived from a graph state on three qubits.

\subsection{Formal definition}\label{sec:formal}
Our setting for applying persistent homology to quantum information is a set $\hal\coloneqq \{A_{1},\ldots,A_{n}\}$ of disjoint quantum systems (parties) to which we associate an \textit{abstract simplicial complex} $\Delta$, which in this work we take as $\Delta = 2^{\hal} \setminus \varnothing$ unless specified otherwise (note that $|\Delta|$ is exponential in the number of parties).
Given a state $\rho \in \mathcal{D}(\hilb_{\hal})$ we consider a functional $F_\rho\colon \Delta \to \mathbb{R}$ that is monotonic under partial trace: \begin{align}
	F(J)_{\rho} \le F(K)_{\rho} \quad \text{if $J \subseteq K$ for $J,K \in \Delta$.}
	\label{eq:monotonicity}
\end{align}
The functional chosen in this work is a deformed version of the total correlation of \eqref{eq:total-correlation}, which we define in Section~\ref{sec:functionals} below.
We also comment on possible alternative choices in Section~\ref{sec:conclusions}.
For a given a state $\rho \in \mathcal{D}(\hilb_{\hal})$ and functional $F_\rho\colon \Delta \to \mathbb{R}$ satisfying the monotonicity condition \eqref{eq:monotonicity}, we define a filtration of simplicial complexes $\lbrace G(\eps)_\rho\colon \eps\in\mathbb{R}\rbrace$ by 
\begin{align}
	G(\varepsilon)_{\rho} = \{J \in \Delta \colon F(J)_{\rho} \le \varepsilon\},
	\label{eq:filtration}
\end{align} 
which satisfies $G(\varepsilon)_{\rho} \le G(\varepsilon')_{\rho}$ for $\varepsilon \le \varepsilon'$. 
Given a field $K$ (in this work we take $K \cong \mathbb{Z}_2$) we compute the \textit{simplicial homology} $H_{k}(G(\varepsilon))$ as follows.
The chain group $\cC_k$ is defined as the group of simplicial $k$-chains of the form $\sum_{j=1}^{n}k_j \sigma_j$, where $k_j \in K$ and $\sigma_j$ is an (oriented) $k$-simplex. 
The boundary operator $\partial_k\colon \cC_k \to \cC_{k-1}$ is a homomorphism acting as $\partial_k (\sigma) = \sum_{i=0}^k (-1)^i(v_0,\ldots,\hat{v}_i,\ldots,v_k)$, where $v_i$ are the vertices of $G(\varepsilon)$ and $\hat{v}_i$ denotes omission. 
We have $\partial_{k}\circ\partial_{k+1} = 0$ and hence $\text{im } \partial_{k+1} \leq \ker \partial_{k}$.
The homology group $H_k \cong \ker \partial_{k} / \text{im } \partial_{k+1}$ identifies the representative chain elements $c_k \in \cC_{k}$ that cannot be written as the boundary of a chain element $c_{k+1} \in \cC_{k+1}$.
Its elements are formally equivalence classes of $k$-chains. 

The inclusion mapping at the level of simplicial complexes induces a homomorphism at the level of the homology groups: $f_{\ast}^{kl}\colon H_\ast(G(\varepsilon_k)) \to H_\ast(G(\varepsilon_l))$, satisfying (1) $f_{\ast}^{kk} $ equates to the identity map and (2) $f_{\ast}^{km} = f_{\ast}^{lm}\circ f_{\ast}^{kl}$ for all $k\le l \le m$. 
This collection of homology groups and homomorphisms indexed by real numbers is termed the \textit{persistence module} $P$ \cite{zomorodian2005persistent}.\footnote{ 
This construction can also be understood in terms of category theory: Define an $\mathbb{R}$-indexed \textit{sublevel set} filtration functor $G\colon\mathbb{R} \to \text{\textbf{Simp}}$ by $G(\varepsilon)_{\rho} = \{J \in \Delta | F(J)_{\rho} \le \varepsilon\} $, which satisfies $G(\varepsilon)_{\rho} \le G(\varepsilon')_{\rho}$ when $\varepsilon \le \varepsilon'$. 
Here, \textbf{Simp} denotes the category of simplicial complexes. 
Denoting by $H_{k}$ the $k$-th homology functor with coefficients in a field $K$, and by \textbf{Vec} the category of vector spaces over $K$, the functor $P_{k} \equiv H_{k}G\colon \mathbb{R} \to \text{\textbf{Vec}}$ can be identified with the persistence module defined above.}
In particular, $P_{k}(\varepsilon)_{\rho}$ is the $k$-th simplicial homology group for the subcomplex $G(\varepsilon)_{\rho}$. 
The $k$-th Betti number $\beta_{k}(\varepsilon)_{\rho}$ for $\varepsilon \in \mathbb{R}$ is defined as $\rk P_{k}(\varepsilon)_{\rho}$.
 
Under a mild finite-type assumption on the chain groups (which is always satisfied for the complexes considered in this work), the persistence module can be uniquely represented by a finite multiset of intervals in $\mathbb{R}$ termed a \textit{barcode} \cite{zomorodian2005persistent}.
This barcode representation is described by a set of tuples $\{(b,d,k,m)\}$, where $(b,d) \subseteq \mathbb{R}$ denotes an interval, $k$ the homological dimension, and $m$ indexing possible degeneracies. 
The interval $(b,d)$ functionally corresponds to a pair of simplices $(J,K)$ that create and destroy a homology group, respectively.

Informally, a barcode encodes intervals of $\mathbb{R}$ over which homology groups \textit{persist}. 
A barcode can equivalently be represented as a \textit{persistence diagram} $\text{PD}_k$, which depicts the intervals $(b,d)$ from a fixed homological dimension $k$ as points in $\mathbb{R}^{2}$. 
For a persistence diagram we refer to the abscissa and ordinate as the ``birth'' and ``death'' times, respectively. 
The distance of a point from the diagonal $y=x$ is proportional to what is termed the \textit{lifetime} $l\coloneqq |d-b|$, serving as a crude quantifier of the ``importance'' of a homology groups: the larger the lifetime, the more ``global'' a topological feature. 
We depict an example barcode of a persistence complex in Fig.~\ref{fig:example}. We note that persistence diagrams often include infinite intervals of the form $(b,\infty)$; in practice, these bars are either ignored or addressed via reduced or relative homology, as discussed below.

\begin{figure}
	\centering
	\includegraphics[width=0.9\textwidth]{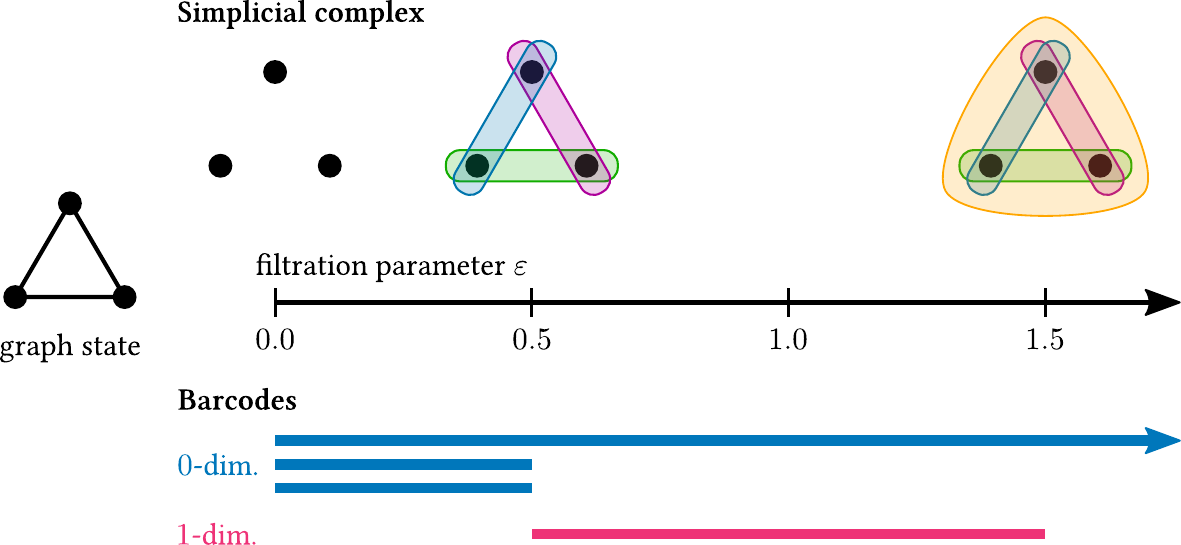}
	\caption{A simple example of a persistence complex induced by the functional $C_2$ defined in~\eqref{eq:deformed-total-correlation}, built from the graph state corresponding to the complete graph on three vertices. The top row shows the simplices that are added to the complex at filtration values $\eps=\lbrace 0,0.5,1.5\rbrace$. The bottom row shows the corresponding barcodes computed from the persistence complex. Evidently, a $1$-dimensional hole persists for $0.5\leq\eps\leq 1.5$.
		At $\eps=1.5$ the final $3$-simplex is added to the complex, which closes the hole. The top $0$-dimensional barcode (blue with arrow) corresponds to the connected component of all vertices which persists for all times.}
	\label{fig:example}
\end{figure}

\subsection{Topological Summaries}\label{sec:topsumm}
The methodology formalized here is a ``persistent homology pipeline'' that maps quantum states to persistence modules $P$, uniquely represented (for the single-parameter case) as a persistence diagram (measure) $\text{PD}_i$ or barcode. 
In practical applications, persistence diagrams can be further analyzed using a so-called \textit{topological summary} (TS) $T\colon P \to \mathbb{R}$. This typically involves the integration of a suitable function with respect to the continuous filtration parameter $\eps$, one important example of which is the \textit{integrated Betti number} $\mathfrak{B}_{k}$:
\begin{align}
    \mathfrak{B}_{k}(\varepsilon) &= \int_{0}^{\varepsilon} \beta_{k}(\varepsilon')d\varepsilon'.
\end{align}
In plain language, $\mathfrak{B}_k(\varepsilon)$ is the integral of the Betti number $\beta_k(\eps)$ defined in \Cref{sec:formal} up to filtration value $\varepsilon$. 
The \textit{total persistence} is given by the sum of the integrated Betti numbers:
\begin{align}
	\mathfrak{B}(\varepsilon) \coloneqq \sum_{k=0}^{\infty}\mathfrak{B}_{k}(\varepsilon).
\end{align}
As we soon demonstrate, the most interesting topological summary in this work is given by the \textit{integrated Euler characteristic} (IEC) \cite{bobrowski2012euler}:
\begin{align}
    \mathfrak{X}(\varepsilon)\coloneqq \sum_{k=0}^{\infty}(-1)^{k}\mathfrak{B}_{k}(\varepsilon).
\end{align}
As the name suggests, the IEC is the integration of the \textit{Euler characteristic} with respect to the filtration.

The quantities $\mathfrak{B}(\eps)$ and $\mathfrak{X}(\eps)$ defined above generally diverge in the limit $\varepsilon \to \infty$ due to the presence of persistence intervals of the form $(b,\infty)$. 
These bars arise naturally in dimension $k=0$, as the connected components of $\Delta$ persist indefinitely. 
To deal with this divergence, we use the \textit{reduced} persistent homology, obtained by defining an augmented chain complex 
\begin{align}
    \ldots \xrightarrow[]{\partial_{1}} \cC_{1} \xrightarrow[]{\partial_{0}} \cC_{0}  \xrightarrow[]{\varphi} \mathbb{Z} \to 0,
\end{align}
such that $\varphi(\sum_{k=1}^{m}\sigma_{k}) = m \mod 2$ when choosing $\mathbb{Z}_2$ as the field of coefficients. 
Essentially, the result of this definition is a redefinition of the Betti numbers as 
\begin{align}
    \widetilde{\beta}_{k}(\varepsilon) \to \begin{cases}
        \beta_{k}(\varepsilon), & k > 0 \\
        \max(\beta_{k}(\varepsilon)- 1, 0), & k = 0.
    \end{cases}
\end{align}

We use tilde notation to indicate the reduced homology, e.g., $\widetilde{\mathfrak{B}}_k$ denotes the reduced integrated Betti number, and $\widetilde{\mathfrak{X}}(\eps)$ the reduced integrated Euler characteristic.
The reduced homology is especially useful if $F(K)_{\rho}=0$ for $K \in V(\Delta)$ (i.e., the vertex set $V(\Delta)$ appears at $\varepsilon = 0$ in the filtration), which will be the case for our choice of function discussed in Section~\ref{sec:functionals} below.
In such cases, the removal of the zero-dimensional infinite bar $(0,\infty)$ (corresponding to the connected component in $\Delta$) does not alter any topological summary apart from disregarding the infinite addend.

\section{Main results}\label{sec:main-results}

\subsection{Choice of Functionals}\label{sec:functionals}
Recall from \eqref{eq:filtration} that the filtration $\lbrace G(\eps)_\rho\colon \eps\in\mathbb{R}\rbrace$ is defined in terms of the choice of field $K$ (typically taken to be $\mathbb{Z}_2$) used to compute homology, the complex $\Delta$ (in this paper we choose $2^{\hal}\setminus \varnothing$), and a functional $F_\rho\colon\Delta\to\mathbb{R}$. 
In this work we choose as $F$ a $q$-deformed version of the total correlation measure defined in \eqref{eq:total-correlation}.
For a quantum state $\rho\in\cD(\cH_\cA)$ and a subset $J\subseteq \cA$, the \emph{$q$-deformed total correlation} for $q>1$ is defined as
\begin{align}\label{eq:deformed-total-correlation}
	C_{q}(J)_{\rho} =  \sum_{v \in J}S_{q}(v)_{\rho} - S_{q}(J)_{\rho},
\end{align}
where the sum runs over all vertices $v$ in the simplex $J$, and $S_q$ is the Tsallis entropy defined in~\eqref{eq:tsallis-entropy}.
In the limit $q \to 1$, Eq.~\eqref{eq:deformed-total-correlation} gives the total correlation measure from \eqref{eq:total-correlation},
\begin{align}\label{eq_rel_ent_filt}
	C(J)_{\rho}\coloneqq \lim_{q\to 1} C_{q}(J)_{\rho} = D\Big(\rho_J\|\bigotimes\nolimits_{v \in J}\rho_v\Big).
\end{align} 
As an example, for $J=\{A,B\}$ we have $C(J) = I(A:B) = S(A)+S(B)-S(AB)$, the mutual information of the quantum state $\rho_{AB}$.
For $J=\{A,B,C\}$, we obtain
\begin{align}
	C(\{A,B,C\})_{\rho} = D(\rho_{ABC} \| \rho_{A}\otimes\rho_{B}\otimes\rho_{C}) = S(A)+S(B)+S(C)-S(ABC).
\end{align} 
The total correlation $C(J)_{\rho}$ has a pleasing operational interpretation as the total amount of local noise that the individual parties in $J$ have to apply one by one in order to decouple their systems from the rest \cite{groisman2005correlations}.

The case $q=2$ corresponds to using the linear entropy $S_{2}(J)_{\rho}$ in Eq.~\eqref{eq:deformed-total-correlation}. 
We will see in Section~\ref{sec:correlation-measures} below that this measure is intimately related to an entanglement measure called the $n$-tangle \cite{coffman2000distributed,wong2001multiparticle}.

We first show that $C_{q}(J)_{\rho}$ is monotonic under partial trace and thus determines a proper sublevel set filtration for all $q\geq 1$, as explained in Section~\ref{sec:ph}:
\begin{lemma}
	Let $\rho\in\cD(\cH_\cA)$ and $q\geq 1$.
	If $I\subseteq J\subseteq \cA$, then
	\begin{align}
		C_{q}(I)_{\rho} \leq C_q(J)_{\rho}.
	\end{align}
\end{lemma}
\begin{proof}
	The Tsallis entropy $S_q$ is subadditive for $q>1$ \cite{audenaert2007tsallis}, which shows that
	\begin{align}
		S_q(J) &\leq S_q(I) + S_q(J\setminus I) \leq S_q(I) + \sum_{v\in J\setminus I} S_q(v).\label{eq:subadditivity}
	\end{align}
	We then have
	\begin{align}
		C_q(J) - C_q(I) &= \sum_{x\in J} S_q(x) - S_q(J) - \sum_{y\in I} S_q(y) + S_q(I)\\
		&= \sum_{v\in J\setminus I} S_q(v) + S_q(I) - S_q(J) \geq 0
	\end{align}
	by \eqref{eq:subadditivity}, which proves the claim for $q>1$.
	For $q=1$, we may use the subadditivity $S(AB)\leq S(A)+S(B)$ of the von Neumann entropy directly.
\end{proof}

It is evident from definition \eqref{eq:deformed-total-correlation} that the $q$-deformed total correlation $C_q$ is invariant under local unitaries (LU): Let $U = U_1 \otimes \dots\otimes U_n$ for some local unitaries $U_i$ acting on the $i$-th system $A_i$, respectively.
Then for any $J\subseteq \cA = \lbrace A_1,\dots,A_n\rbrace$ and an arbitrary state $\rho\in\cD(\cH_\cA)$ we have $C_q(J)_\rho = C_q(J)_{U\rho U^\dagger}$, which follows from the unitary invariance of the Tsallis entropy $S_q$.
Hence, LU-equivalent states have identical barcodes, the converse of which we record for future reference.
\begin{lemma}\label{lem:different-barcodes-LU-inequivalent}
	If $\rho$ and $\sigma$ are two quantum states with different barcodes with respect to the deformed total correlation, then $\rho$ and $\sigma$ are LU-inequivalent.
\end{lemma}
We note that a barcode derived from the deformed total correlation does \emph{not} identify a quantum state, say, up to LU-equivalence.
This follows from the facts that the deformed total correlation $C_q$ is computed from the eigenvalues of the marginals of a quantum state, and that there are pairs of LU-inequivalent states whose marginals all have identical spectra.
A particular example of this is discussed at the end of Section~\ref{sec:correlation-measures}.
Functions of the spectra of marginals of a multipartite quantum state are examples of so-called \emph{LU invariants}.
In principle, LU invariants can be used to determine a quantum state up to LU equivalence, but this quickly becomes intractable for a growing number of parties \cite{rains2000polynomial,grassl1998local}.
Even if one relaxes LU equivalence to SLOCC equivalence (see Section~\ref{sec:correlation-measures}), a complete classification of multipartite entanglement is intractable except for small system sizes~\cite{Dur_Vidal_Cirac_2000,verstraete2002four,briand2003,osterloh2005constructing,luque2005algebraic,osterloh2006entanglement,klyachko2007dynamical}.
This complexity renders a complete characterization of multipartite entanglement generally infeasible, and one has to resort to approaches designed to capture certain aspects of entanglement.

Our approach based on persistent homology links entanglement properties of a multipartite state to topological features of a simplicial complex constructed from the state and its marginals.
This process can be understood as a certain ``discretization'' of entanglement features that is nevertheless stable with respect to perturbations in the following sense.
The Tsallis entropy $S_q$ is continuous with respect to the trace distance $\|\rho-\sigma\|_1\coloneqq \Tr |\rho-\sigma|$ (where $|X|\coloneqq \sqrt{X^\dagger X}$) on quantum states (see, e.g., \cite{raggio1995entropies}).
This also implies continuity for the $q$-deformed total correlation $C_q$ defined in \eqref{eq:deformed-total-correlation} above.
We then have the following stability property of the persistence diagrams \cite{Cohen-Steiner2007,Skraba_Thoppe_Yogeshwaran_2020}:
Given two quantum states $\rho,\sigma$ that are close in trace distance, the persistence diagrams of the persistence complexes of $\rho,\sigma$ defined in terms of $C_q$ are close in the bottleneck distance.
In particular, any of the topological summaries defined in \Cref{sec:topsumm} is continuous with respect to the quantum state defining the persistence complex.

Finally, we define a $q$-deformed version of the interaction information introduced in \eqref{eq:interaction-information}.
For a quantum state $\rho\in\cD(\cH_\cA)$ and $q>1$,
\begin{align}
	I_{q}(\hal)_\rho \coloneqq \sum_{J\subseteq \hal}(-1)^{|J|-1}S_q(J)_\rho, 
	\label{eq:deformed-interaction-information}
\end{align} 
where we set $S_q(\varnothing)_\rho = 0$.
We will see in the next section that the $q$-deformed interaction information $I_q$ arises as the topological summary of a persistence complex built from the $q$-deformed total correlation $C_q$.

We note that the functional in Eq.~\eqref{eq:deformed-interaction-information} corresponds to a particular parameterization of a ``state index" function proposed in the recent work \cite{mainiero2019homological} that also leveraged homological tools to understand multipartite entanglement. 
Therein, the state index function was proposed as the ``Euler characteristic'' of some associated non-commutative geometry attached to a density state \cite{mainiero2019homological}. As we demonstrate below, this quantity does indeed correspond to an \textit{integrated} Euler characteristic in the context of \textit{persistent} homology.

\subsection{Correlation measures as topological summaries of persistence complexes}\label{sec:correlation-measures}

We now show that multipartite entanglement measures and inequalities may arise as topological summaries of persistence diagrams. 
% In particular, consider the case where $\Delta = \Delta_{\hal}= 2^{\hal}$, the powerset of $ \hal$.
Recalling that we set $\Delta = 2^{\hal}\setminus \varnothing$, we clearly have $\beta_{i>0}(\Delta) = 0 = 1-\beta_{0}(\Delta)$. 
By computing the reduced persistent homology using the $q$-deformed total correlation functional $C_q(J)_{\rho}$ defined in \eqref{eq:deformed-total-correlation} above, we obtain our first main result:
\begin{theorem}
	\label{thm:interaction-information}
	Let $\rho\in\cD(\cH_\cA)$ be a quantum state on $\cA=\lbrace A_1,\dots,A_n\rbrace$ and let $P$ be the corresponding persistence module defined via \eqref{eq:filtration} in terms of the $q$-deformed total correlation $C_q$ in \eqref{eq:deformed-total-correlation}.
	Then the reduced integrated Euler characteristic $\widetilde{\mathfrak{X}}_{q}(\infty)$ of this complex is equal to the $q$-deformed interaction information $I_q$ in \eqref{eq:deformed-interaction-information}:
\begin{align}\label{eq:IEC-equals-total_correlation}
    \widetilde{\mathfrak{X}}_{q}(\infty) = \sum_{J \subseteq \hal}(-1)^{|J|-1}S_{q}(J)_{\rho} \eqqcolon I_{q}(\hal)_{\rho}.
\end{align} 
\end{theorem}

\begin{proof}
	\newcommand{\emax}{{\eps_{\mathrm{max}}}}
	Let us denote by $\emax$ the largest value of the filtration parameter at which the number of barcodes changes, and let $n_k(\eps)$ be the number of $k$-simplices existing at time $\eps$.
	The Euler characteristic $\mathfrak{X}(\eps)$ for fixed $\eps$ can be expressed in terms of the $n_k$ as \cite[Thm.~2.44]{MR1867354}
	\begin{align}
		\mathfrak{X}(\eps) = \sum_{k=0}^{\infty} (-1)^k n_k (\eps).
	\end{align}
	For the integrated Euler characteristic we thus obtain
	\begin{align}
		\mathfrak{X}(\emax) = \int_0^\emax \sum_{k=0}^\infty (-1)^k n_k(\eps) \,d\eps = \sum_{k=0}^\infty (-1)^k  \int_0^\emax n_k(\eps) \,d\eps.\label{eq:euler-char}
	\end{align}
	Recall that, for given $\eps$, a simplex $J$ is added to the complex if $C_q(J)\leq \eps$.
	Any simplex $J$ is thus ``born'' at time $C_q(J)$ and contributes $(\emax - C_q(J))$ to the integral.
	Hence, denoting by $\Delta_k$ the (complete) set of $k$-simplices, we have
	\begin{align}
		\int_0^\emax n_k(\eps) \,d\eps = \sum_{J\in\Delta_k} (\emax - C_q(J)) = \binom{n}{k+1} \emax - \sum_{J\in\Delta_k}C_q(J).
	\end{align}
	Substituting this in \eqref{eq:euler-char} gives
	\begin{align}
		\mathfrak{X}(\emax) &= \emax \sum_{k=0}^\infty(-1)^k \binom{n}{k+1} - \sum_{k=0}^\infty (-1)^k \sum_{J\in\Delta_k} C_q(J)\\
		&= \emax - \sum_{J\in\Delta} (-1)^{|J|-1} C_q(J),\label{eq:intermediate}
	\end{align}
    where we used the binomial identity $\sum_{j=0}^n (-1)^j \binom{n}{j}=0$ in the last equality.
    
	We calculate the second term in \eqref{eq:intermediate} using the definition $C_q(J)=\sum_{x\in J} S_q(x) - S_q(J)$ of the $q$-deformed total correlation:
	\begin{align}
		- \sum_{J\in\Delta} (-1)^{|J|-1} C_q(J) &= - \sum_{J\in\Delta} (-1)^{|J|-1} \sum_{x\in J} S_q(x) + \sum_{J\in\Delta} (-1)^{|J|-1} S_q(J)\\
		&= \sum_{J\in\Delta} (-1)^{|J|-1} S_q(J).
	\end{align}
	The second equality follows from $\sum_{J\in\Delta} (-1)^{|J|} \sum_{x\in J} S_q(x) = 0$, which can be proved using induction over the number of vertices in $\Delta$.
	In summary, we obtain the following formula for the integrated Euler characteristic:
	\begin{align}
		\mathfrak{X}(\emax) = \emax  + \sum_{J\in\Delta} (-1)^{|J|-1} S_q(J).\label{eq:IEC-nonreduced}
	\end{align}
 
    Taking the reduced homology amounts to omitting in $\mathfrak{X}_q(\eps_{\mathrm{max}})$ the contribution of one of the $0$-simplices.
    This removes the $\emax$ term in \eqref{eq:IEC-nonreduced}, and so we arrive at 
    \begin{align}
    \widetilde{\mathfrak{X}}_{q}(\infty) = \sum_{J \subseteq \hal}(-1)^{|J|-1}S_{q}(J)_{\rho},
    \end{align}
    which concludes the proof.
\end{proof}

We note that the result of \Cref{thm:interaction-information} does not depend on the choice of field $K$ made in \Cref{sec:formal}.
The proof of \Cref{thm:interaction-information}, in particular \eqref{eq:intermediate}, reveals an explicit formula for the reduced integrated Euler characteristic $\widetilde{\mathfrak{X}}_F(\infty)$ of a persistence complex defined in terms of an \emph{arbitrary} functional $F$.
This may help in identifying other functionals $F$ that give rise to persistence complexes with interesting topological summaries such as the reduced integrated Euler characteristic.
We thus record this observation as a stand-alone result:
\begin{theorem}
	Let $\rho \in \mathcal{D}(\hilb_{\hal})$ be a quantum state and let $F_\rho\colon 2^{\hal} \setminus \varnothing \to \mathbb{R}$ be a functional satisfying monotonicity under partial trace \eqref{eq:monotonicity}.
	Then the reduced integrated Euler characteristic $\widetilde{\mathfrak{X}}_{F_\rho}(\infty)$ of the persistence complex defined via \eqref{eq:filtration} is given by the following formula:
	\begin{align}
		\widetilde{\mathfrak{X}}_{F_\rho}(\infty) = \sum_{J\subseteq \hal} (-1)^{|J|} F(J)_\rho.
	\end{align}
\end{theorem}

Going back to Theorem~\ref{thm:interaction-information}, we note that Eq.~\eqref{eq:IEC-equals-total_correlation} identifies a correspondence between the $q$-deformed total correlation $C_q(J)_{\rho}$ and the total mutual information $I_q(J)_{\rho}$, given by the following identities:
\begin{align}
    I_q(J)_{\rho} =  -\sum_{K\le J}(-1)^{|K|}C_q(K)_{\rho},\\
    C_q(J)_{\rho} = \sum_{K \le J}(-1)^{|K|}I_q(K)_{\rho},
\end{align}
where we have assumed $I_q(K)_{\rho} = 0$ if $|K|  \le 1$ (our convention is $\rho_{K} = 1$ if $|K|=0$). 

Somewhat surprisingly, the choice $q=2$ when $\hal$ is a set of $n$ qubit labels (with $n$ even) yields an important entanglement monotone: the Minkowski length of the generalized Bloch vector, which is sometimes called the Stokes tensor.
This quantity corresponds to the $n$-tangle \cite{coffman2000distributed,wong2001multiparticle}, which is defined for a state $\rho$ on an even number of qubits as
\begin{align}
	\tau_n(\rho) = \Tr(\rho\sigma_2^{\otimes n}\rho^* \sigma_2^{\otimes n}),
\end{align}
where $\sigma_2=\left(\begin{smallmatrix}0&-i\\i&\phantom{-}0\end{smallmatrix}\right)$ is the Pauli $Y$-matrix and $\rho^*$ denotes the complex conjugate of $\rho$ taken entry-wise with respect to the standard basis used in the definition of $\sigma_2$.

Define now the generalized $n$-qubit Bloch vector $\left(Q_{(i_{1},\ldots, i_{n})}\right)_{i_1,\dots,i_n}$ via
\begin{align}
    Q_{(i_{1},\ldots, i_{n})} = \Tr\left(\rho \left(\sigma_{i_{1}}\otimes \dots \otimes \sigma_{i_{n}}\right)\right), \, i_{j} = 0,1,2,3, \label{eq:stokes-parameters}
\end{align} 
where $\sigma_{0} = I_{2}$ is the identity, $\sigma_{\mu}$ for $\mu =1,2,3$ are the three Pauli matrices, and $\frac{1}{2}\Tr(\sigma_{\mu}\sigma_{\nu}) = \delta_{\mu\nu}$ \cite{Jaeger_Sergienko_Saleh_Teich_2003,Jaeger_Teodorescu-Frumosu_Sergienko_Saleh_Teich_2003}. 
An invertible SLOCC transformation corresponds to a local action of the special linear group $SL(2,\mathbb{C})$ \cite{Dur_Vidal_Cirac_2000}.
This induces a transformation of the generalized Bloch vector coordinates in \eqref{eq:stokes-parameters} via local operations from the isomorphic group $O_{0}(1,3)$, defined as the proper Lorentz group that leaves invariant the Minkowski length of the Bloch vector \cite{Jaeger_Teodorescu-Frumosu_Sergienko_Saleh_Teich_2003}:
\begin{align}
    Q_{(n)}^{2} \equiv \frac{1}{2^{n}}\sum_{\iota \in I}(-1)^{|\iota|}Q_{\iota}^{2}.\label{eq:minkowski-length}
\end{align} 
Here, $I$ is the set of multi-indices $(i_{1},\ldots,i_{n})$ and for $\iota \in I$ we define $|\iota|$ to be the number of qubits for which $\sigma_{\iota}\coloneqq \bigotimes_{j \in [n]}\sigma_{\iota_{j}}$ acts non-trivially. 
As an example, if $n=4$ and $\iota = (0,2,3,0)$ then $|\iota| = 2$. 

While for $n$ odd the Minkowski length \eqref{eq:minkowski-length} vanishes on pure states $\rho$, for $n$ even the Minkowski length of the Bloch vector exactly corresponds to the $n$-tangle \cite{Jaeger_Teodorescu-Frumosu_Sergienko_Saleh_Teich_2003}:
\begin{align}
	Q_{(n)}^{2} = \tau_n(\rho).
\end{align}
The generalized Bloch vector also gives a convenient way to express the purity, and therefore the linear entropy, of reduced density matrices. In particular we can write \cite{Jaeger_Teodorescu-Frumosu_Sergienko_Saleh_Teich_2003}
\begin{align}
    S_2(J)_\rho = 1-\frac{1}{2^{|J|}}\sum_{\iota \in I_J}Q_\iota ^2,
\end{align} 
where $I_J \subseteq I$ corresponds to the multi-indices in $I$ with non-trivial support only in the subset $J$, along with the trivial identity Pauli string. 
Continuing our example above with $n=4$ qubits, the set $I_J$ for the two middle qubits (assuming for simplicity a one-dimensional chain) has the form $\iota = (0,i,j,0)$.

The alternating sum form of the Minkowski length $Q_{(n)}^{2}$ in \eqref{eq:minkowski-length} resembles that of Eq.~\eqref{eq:IEC-equals-total_correlation}. 
In fact, this similarity is not just coincidental, as observed in our next main result:
\begin{theorem}
	\label{thm:n-tangle}
	Let $\rho\in\cD(\cH_\cA)$ be a quantum state on $n$ qubits (with $n$ even), and let $P$ be the corresponding persistence module defined via \eqref{eq:filtration} in terms of the total correlation $C_2$ in \eqref{eq:deformed-total-correlation}.
	Then the reduced integrated Euler characteristic $\widetilde{\mathfrak{X}}_{2}(\infty)$ of this complex is equal to the Minkowski length of the generalized Bloch vector, and hence also equal to the $n$-tangle $\tau_n$:
\begin{align}
     \widetilde{\mathfrak{X}}_{2}(\infty) = I_{2}(\hal)_{\rho} =Q_{(n)}^{2}=  \tau_n(\rho).
\end{align} 
\end{theorem}

\begin{proof}
Starting from the definition of $I_2$ in \eqref{eq:deformed-interaction-information}, we have 
\begin{align}
\begin{aligned}
	    I_2(\hal)_\rho &= \sum_{J \in \Delta }(-1)^{|J|-1}S_2(J)_\rho\\
	    &=  \sum_{J\in \Delta}(-1)^{|J|-1}\left(1-\frac{1}{2^{|J|}}\sum\nolimits_{\iota \in I_J}Q_\iota ^2\right)\\
     &= 1 + \sum_{J \in \Delta}(-1)^{|J|}\frac{1}{2^{|J|}}\sum_{\iota \in I_J}Q_\iota ^2.
\end{aligned}
\end{align}
We exchange the two summations by summing over $\iota \in I$, weighted by the number of sets $J \in \Delta$ that are supersets of $\text{supp}(\iota)$. We then have 
\begin{align}
       I_2(\hal)_\rho = 1 + \sum_{\iota \in I}Q_{\iota}^2 \sum_{J \supseteq \text{supp}(\iota), \, J \ne \varnothing} (-2)^{-|J|} 
       =\frac{1}{2^n}\sum_{\iota \in I}(-1)^{|\iota|}Q_{\iota}^2.
\end{align}
In the final equality we use the fact that $\sum_{J \supseteq \text{supp}(\iota), J \ne \varnothing } (-2)^{-|J|} = (-1)^{|\iota|}/2^n$ for $|\iota| >0$. For the special case of $|\iota| = 0$, we have 
\begin{align}
   \sum_{J \supseteq \text{supp}(\iota), \, J \ne \varnothing} (-2)^{|J|} =  \sum_{J \in \Delta} (-2)^{|J|} = 2^{-n} - 1,
\end{align} which yields the result as claimed.
\end{proof}

%That the Minkowski length amounts to a $q$-deformed version of the interaction information has to our knowledge not been established in the literature. 

The fact that $\widetilde{\mathfrak{X}}_2(\infty)$ is equal to the $n$-tangle shows that $\widetilde{\mathfrak{X}}_2(\infty)$ is monotonic on average under local operations and classical communications (LOCC). 
To wit, any LOCC operation $\Lambda$ on $\hilb_{\hal}$ has a Kraus decomposition 
\begin{align}
	\Lambda(\rho) = \sum\nolimits_{i}K_{i}(\rho)K_{i}^{\dagger},
\end{align} 
where $K_{i}\coloneqq \bigotimes_{A_{i}\in \hal} K_{A_{i}}$ and $\sum_{i}K_{i}^{\dagger}K_{i}=1_{\hilb_{\hal}}$, the identity operator on $\hilb_{\hal}$.
This operation yields post-measurement states $\sigma_{i}= K_{i}\rho K_{i}^{\dagger}/\Tr(K_{i}\rho K_{i}^{\dagger})$ occurring with probability $p_{i} = \Tr K_{i}\rho K_{i}^{\dagger}$.
An entanglement measure $F$ is now called an entanglement monotone under LOCC \cite{Vidal2000} if
\begin{align}
	F(J)_\rho \ge \sum\nolimits_{i}p_{i}F(J)_{\sigma_i}.
\end{align}
By direct analogy, we say a topological summary (TS) is an entanglement monotone if 
\begin{align}
	T^\rho \ge \sum\nolimits_{i}p_{i}T^{\sigma_i},\label{eq:TS-entanglement-monotone}
\end{align} 
where $T^\rho $ denotes the TS for state $\rho$. 
Theorem~\ref{thm:n-tangle} now immediately yields the following:
\begin{corollary}
	The reduced integrated Euler characteristic $\widetilde{\mathfrak{X}}_{2}(\infty) = I_{2}(\hal)_{\rho}$ of the persistence complex from Theorem~\ref{thm:n-tangle} is an entanglement monotone.	
\end{corollary}
Identifying conditions under which topological summaries serve as entanglement monotones in the sense of \eqref{eq:TS-entanglement-monotone} is an intriguing research direction that we will explore in future work.

Theorem~\ref{thm:n-tangle} answers an open question posed by Eltschka and Siewert~\cite{Eltschka2018distributionof}, in which they defined a quantity $C_D(\psi)$ called \emph{distributed concurrence}.
In our notation, the square of $C_D(\psi)$ can be expressed as
\begin{align}
	C_D^2(\psi) = 2 \sum_{J\subseteq \cA} (-1)^{|J|-1} S_2(J)_\rho.\label{eq:distributed-concurrence}
\end{align}
Both $C_D$ and $C_D^2$ are entanglement monotones in the case of $n$ qubits, and $C_D^2$ is an entanglement monotone if no local dimension exceeds $3$.
Evidently, the expression of $C_D^2$ in \eqref{eq:distributed-concurrence} coincides with the Euler characteristic $\widetilde{\mathfrak{X}}_{2}(\infty) = I_{2}(\hal)_{\rho}$ in Theorem~\ref{thm:n-tangle} (up to an irrelevant factor of $2$).
These quantities are therefore entanglement monotones also in the more general case of an $n$-partite system whose local dimensions are at most $3$.
Moreover, in \cite{Eltschka2018distributionof} the authors noticed that \eqref{eq:distributed-concurrence} is reminiscent of an Euler characteristic and pondered the question whether the distributed concurrence can be assigned a topological meaning.
Here, we answer this question in the affirmative by realizing $C_D^2$ as the Euler characteristic of a topological space defined in terms of the multipartite entanglement structure of a given quantum state.
More precisely, $C_D^2$ arises as the integrated Euler characteristic of the persistence complex of a quantum state defined in terms of the $2$-deformed total correlation $C_2$ as given in \eqref{eq:deformed-total-correlation}.

A practical consequence of Theorem~\ref{thm:n-tangle} is that persistent barcodes provide \textit{finer} information about multipartite entanglement than $\widetilde{\mathfrak{X}}_{2}(\infty) = \tau_n$ alone, since the latter is the topological summary of the topological data encoded in the persistent barcodes.
We illustrate this in the following with two examples, which make use of the software package Javaplex \cite{Javaplex} to compute persistence barcodes.\footnote{MATLAB code to reproduce the results shown here is available at \url{https://github.com/felixled/entanglement_persistent_homology}.}

For the first example, we recall the notion of a \emph{graph state} \cite{hein2006entanglement}:
Let $|+\rangle = (|0\rangle+|1\rangle)/\sqrt{2}$ be the ``plus'' state of a qubit system $\mathbb{C}^2$ with orthonormal basis $\lbrace |0\rangle,|1\rangle\rbrace$, and let $CZ$ denote the controlled $Z$-gate, defined as the two-qubit operation $CZ\coloneqq |0\rangle\langle 0|\otimes \sigma_0 + |1\rangle\langle 1|\otimes \sigma_3$ (with $\sigma_3$ the Pauli-$Z$-operator).
For a simple undirected graph $G=(V,E)$ with vertex set $V$ and edge set $E$, the graph state $|G\rangle\in (\mathbb{C}^2)^{\otimes |V|}$ on $|V|$ qubits is defined as
\begin{align}
	|G\rangle \coloneqq \sum_{e\in E} CZ_e |+\rangle^{\otimes |V|},
\end{align}
where the notation $CZ_e$ means that $CZ$ acts on the two vertices incident with the edge $e\in E$.

For our example we choose $|V|=6$ and let $G_1$ and $G_2$ be the two graphs depicted on the left-hand side of the top and bottom rows of Figure~\ref{fig:comparison-graph-states}, respectively.
Note that the graph state $|G_2\rangle$ is local unitarily equivalent to a GHZ state $|\phi\rangle = (|0\rangle^{\otimes 6} + |1\rangle^{\otimes 6})/\sqrt{2}$ \cite{hein2006entanglement}.
It was shown in \cite{schatzki2022tensor} that a graph state on $n$ qubits has $n$-tangle equal to $1$ if all vertices have odd degree, and $0$ otherwise.
Since both $G_1$ and $G_2$ satisfy the former condition, we have $\tau_n(|G_1\rangle) = 1 = \tau_n(|G_2\rangle)$.

We now consider the persistence complexes for $|G_i\rangle$ induced by the $2$-deformed total correlation $C_2$ from \eqref{eq:deformed-total-correlation}.
The resulting persistence barcodes, computed in MATLAB using the software package Javaplex \cite{Javaplex}, are shown in Figure \ref{fig:comparison-graph-states}.
By Theorem~\ref{thm:interaction-information} and \ref{thm:n-tangle}, the $n$-tangle $\tau_n$ is equal to the reduced integrated Euler characteristic, which in turn equals the alternating sum of the lengths of all barcodes (omitting the infinitely long persisting $0$-dim.~barcode).
One can check using the data in Figure \ref{fig:comparison-graph-states} that we indeed have $\tau_n(|G_1\rangle) = 1 = \tau_n(|G_2\rangle)$ using the barcode data, as predicted by the results of \cite{schatzki2022tensor}.\footnote{For example, for the graph $G_1$ shown in the top row of Figure \ref{fig:comparison-graph-states}, there are $5$ (finite) $0$-dim.~barcodes of length $1/4$ each, $6$ $1$-dim.~barcodes of length $1/4$ and $4$ $1$-dim.~barcodes of length $1/2$, $10$ $2$-dim.~barcodes of length $1/2$, $4$ $3$-dim.~barcodes of length $1/2$ and one $3$-dim.~barcode of length $3/4$, and finally one $4$-dim.~barcode of length~$1$. It follows that
\begin{align}
	\tau_n = 1 = 5\cdot \frac{1}{4} - 4\cdot \frac{1}{2} - 6\cdot \frac{1}{4}+10 \cdot \frac{1}{2}  -4 \cdot \frac{1}{2}- \frac{3}{4}+1.
\end{align}}

However, the barcodes in Figure \ref{fig:comparison-graph-states} are evidently different, which by Lemma~\ref{lem:different-barcodes-LU-inequivalent} implies that $|G_1\rangle$ and $|G_2\rangle$ are LU-inequivalent.
Since LU-equivalence is the same as SLOCC-equivalence for graph states \cite{hein2006entanglement}, Lemma~\ref{lem:different-barcodes-LU-inequivalent} actually implies that $|G_1\rangle$ and $|G_2\rangle$ are in different SLOCC orbits, which is revealed by the rather different topological structures of the underlying persistence complexes.
For example, for the filtration interval $\eps\in[1.25,1.5]$ the graph state $|G_1\rangle$ shows the simultaneous existence of two $2$-dimensional barcodes and a $3$-dimensional barcode, whereas in the persistence complex for $|G_2\rangle$ all $2$-dimensional holes are immediately closed.

\begin{figure}[t]
	\centering
	\includegraphics[width=0.8\textwidth]{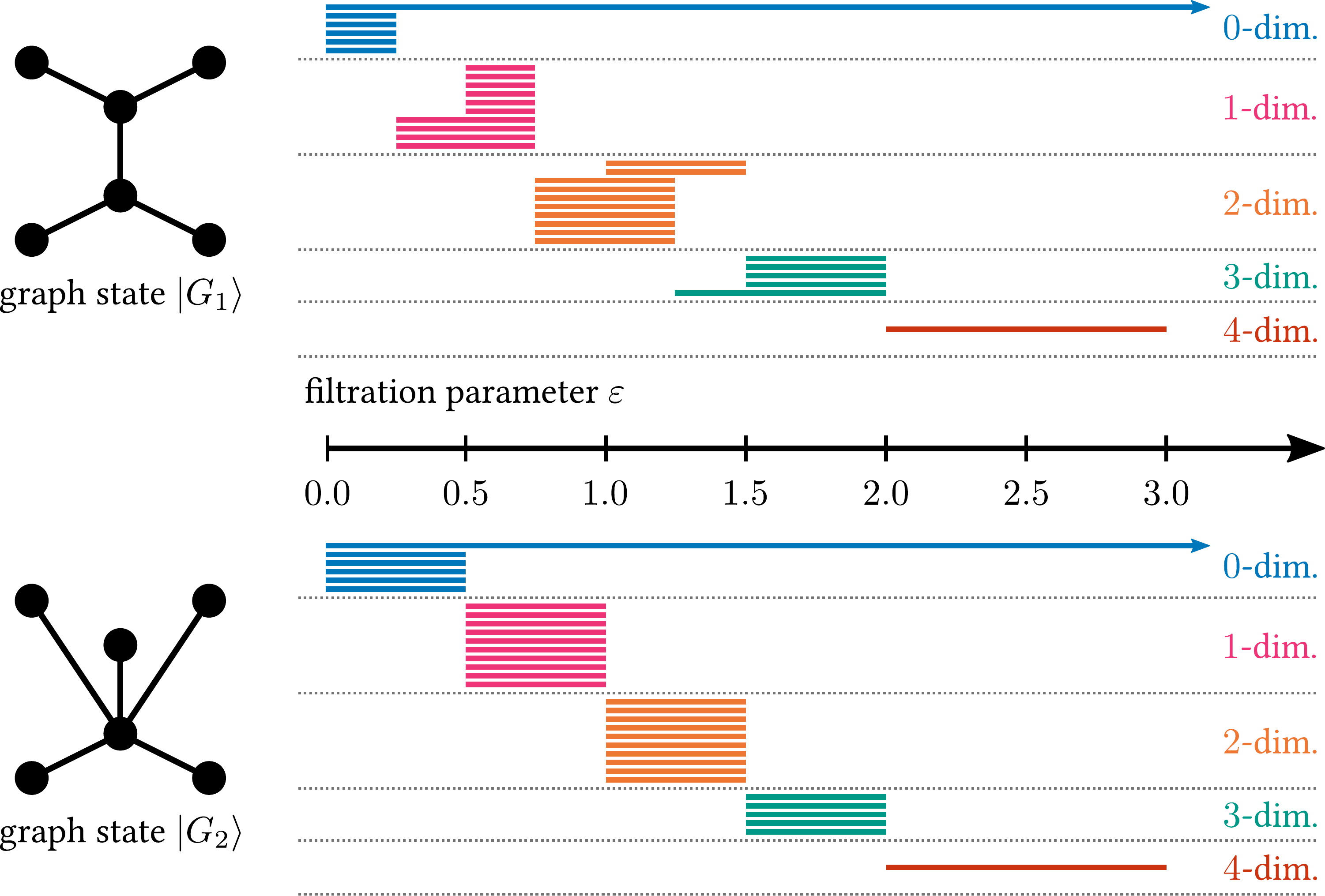}
	\caption{Comparison of the persistence barcodes for two graph states defined in terms of different $6$-vertex graphs, computed with the $2$-deformed total correlation $C_2$ in \eqref{eq:deformed-total-correlation} as the defining functional. 
		In each graph all vertices have odd degree, so that both graph states have $n$-tangle $\tau_n = 1$ \cite{schatzki2022tensor}.
		However, the persistence barcodes reveal a rather different topological structure of the corresponding persistence complexes, reflecting the fact that the two graph states are LU- and SLOCC-inequivalent.}
	\label{fig:comparison-graph-states}
\end{figure}

As another example, we consider the normalized versions of the following two $6$-qubit vectors, defined in terms of the computational basis $\lbrace |0\rangle,|1\rangle\rbrace$ and a parameter $t>0$:
\begin{align}
\begin{aligned}
		|\chi_4\rangle &\propto \frac{1}{t} |111111\rangle + \frac{1}{t}|111100\rangle + t|000010\rangle + t|000001\rangle\\
	|\chi_5\rangle &\propto \frac{\sqrt{2}}{t} |111111\rangle + \frac{1}{t} |111000\rangle + t |000100\rangle + t|000010\rangle + t |000001\rangle.
\end{aligned}
\label{eq:chi}
\end{align}
These states arise from applying the invertible SLOCC-operation $A_t = t|0\rangle\langle 0| + \frac{1}{t} |1\rangle\langle 1|$ to the first qubit of the states $\Xi_4$ and $\Xi_5$ defined in \cite[Sec.~IV]{osterloh2006entanglement}.
The states $\Xi_4$ and $\Xi_5$ both have vanishing $n$-tangle and lie in different SLOCC-equivalence classes \cite{osterloh2006entanglement}.
The same is hence true for the states $\chi_4$ and $\chi_5$ as well, but we note that neither $\chi_4$ and $\chi_5$ are in normal form for $t\neq 1$, since its single-qubit marginals are not all completely mixed.
SLOCC-inequivalence can hence not be inferred just from LU-inequivalence alone \cite{verstraete2003normal}.

We choose $t=4/3$ in \eqref{eq:chi} and compute the persistence complexes of $\chi_4$ and $\chi_5$ with respect to the $2$-deformed total correlation $C_2$ (whose integrated Euler characteristic equals the $n$-tangle).
The resulting barcodes are shown in \Cref{fig:comparison-osterloh-siewert} and reveal the different entanglement structures of these two states, which we take as indicating their SLOCC-inequivalence.

Note that, even though the $n$-tangle is invariant under local $SL(2,\mathbb{C})$-operations~\cite{wong2001multiparticle}, it cannot itself be used to detect SLOCC-inequivalence of two normalized states.
This is because the $n$-tangle is homogeneous of degree 2 (as a function of an operator) and hence changes its value under normalization.
As a result, SLOCC-equivalent normalized $n$-qubit states (with $n$ even) may have different $n$-tangle values.
This can be fixed by taking another homogeneous SLOCC invariant $f(\psi)$ of degree $k$, and defining the functional $g(\psi)\coloneqq \tau_n(\psi)/f(\psi)^{2/k}$ provided that $f(\psi)\neq 0$.
The functional $g$ is still SLOCC-invariant, and in addition invariant under scaling as a ratio of homogeneous functions of the same degree.
It follows that two states $\psi_1$ and $\psi_2$ with $f(\psi_i)\neq 0$ are SLOCC-inequivalent if $g(\psi_1)\neq g(\psi_2)$.

The proof of \Cref{thm:n-tangle} shows that the integrated Euler characteristic changes linearly with a rescaling of the functional $F$ defining the persistence complex.
Hence, for analyzing the SLOCC-equivalence of two given states $\psi_1$ and $\psi_2$ we may apply the construction above: choose a homogeneous SLOCC invariant $f(\psi)$ of degree $k$ with $f(\psi_i)\neq 0$ and redefine the $2$-deformed total correlation as $\widehat{C}_2 \coloneqq C_2/f(\psi)^{2/k}$.
The resulting integrated Euler characteristic is equal to the functional $g$ from above, and the barcodes still encode finer information about the entanglement structure.
For the two states $\chi_4$ and $\chi_5$ in \eqref{eq:chi}, the SLOCC invariant in eq.~(7) of \cite{Gour_Wallach_2013} with $\cA_q = \lbrace 4,5\rbrace$ and $l=6$ is non-zero on each, and can hence be used to rescale the persistence complex as described above.
The resulting barcodes are qualitatively identical to those in \Cref{fig:comparison-osterloh-siewert}.
We conjecture that such rescaled persistence barcodes can be used more generally to detect SLOCC-inequivalence.
However, we caution that the practicality of such a method is limited by the exponential scaling of computing the persistence complex and associated persistence barcodes.\footnote{A simple heuristic approach to speed up the computation of the simplicial complex $G(\eps)$ for a given filtration parameter $\eps$ is the following. Recall that a simplex $J\subseteq \cA$ is added to $G(\eps)$ if $F(J)\leq \eps$, and that $F$ is monotonic under taking partial traces, $F(J)\leq F(K)$ if $J\subseteq K$. 
Hence, in order to build the simplicial complex of an $n$-qubit system (where $n=|\cA|$), it is advantageous to start at the top of the subset lattice and first evaluate $F$ on the $n$ subsets of size $n-1$, then on the $\binom{n}{2}$ subsets of size $n-2$, etc. As soon as $F(K)\leq \eps$ for some $K\subseteq \cA$, we automatically add all subsets $J\subseteq K$ to the complex as well and move to the next subset of size $|K|$.}

\begin{figure}[t]
	\centering
	\includegraphics[width=0.7\textwidth]{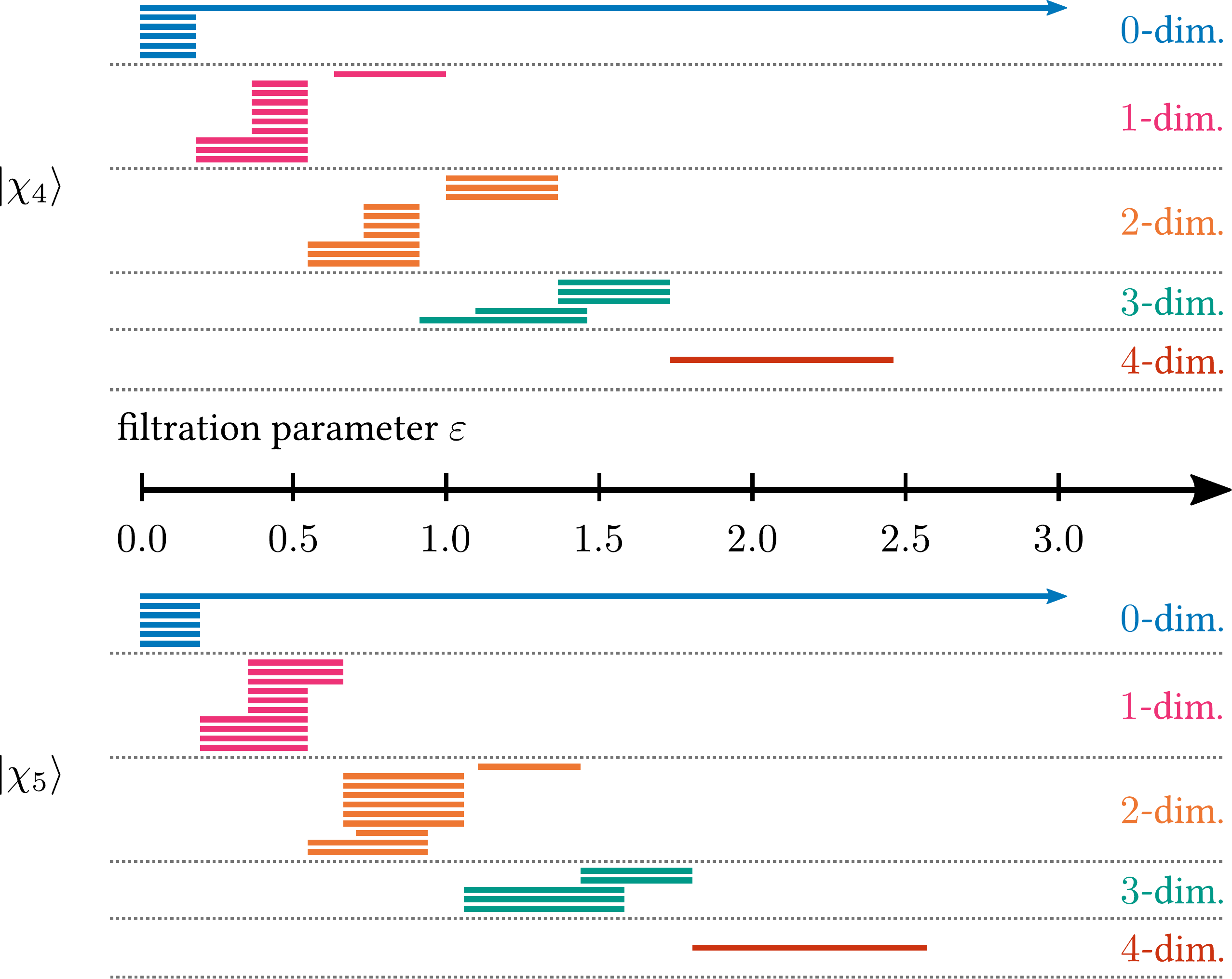}
	\caption{Comparison of the persistence barcodes for the two $6$-qubit states $|\chi_4\rangle$ and $|\chi_5\rangle$ defined in \eqref{eq:chi}, respectively.
		The barcodes are computed with the $2$-deformed total correlation $C_2$ in \eqref{eq:deformed-total-correlation} as the defining functional. 
		The two states are SLOCC-inequivalent \cite{osterloh2006entanglement}, whereas the $n$-tangle vanishes for both states.
		In contrast, the barcodes once again reveal different topological features corresponding to the different entanglement structures.}
	\label{fig:comparison-osterloh-siewert}
\end{figure}

To summarize, the examples in Figures~\ref{fig:comparison-graph-states} and \ref{fig:comparison-osterloh-siewert} suggest that the persistence barcodes convey more information about the entanglement structure than its topological summary, the $n$-tangle, alone.
Intuitively, the integrated Euler characteristic forgets the persistence of the homology groups, while this data is represented in the barcodes.
An interesting question for future research is to determine whether the barcodes themselves, or other functions thereof, have an operational interpretation in terms of entanglement monotones.

We conclude this section by stressing that the persistence complexes considered in this paper are defined in terms of the $q$-deformed total correlation, a function of entropic quantities evaluated on marginals of a multipartite state (see \Cref{sec:functionals}).
Such correlation measures are limited in detecting multipartite entanglement, as the following example illustrates \cite{nezami2020multipartite,walter2023private}: Consider a tripartite system $A_1A_2A_3$ with each system consisting of two qubits, $A_i = A_{i,1}A_{i,2}$ for $i=1,2,3$.
Denote by $|\chi_n\rangle = \frac{1}{\sqrt{2}}(|0\rangle^{\otimes n} + |1\rangle^{\otimes n})$ an $n$-partite GHZ state, and consider the two states 
\begin{align}
	\begin{aligned}
		|\psi_1\rangle_{A_1A_2A_3} &= |\chi_2\rangle_{A_{1,1}A_{2,1}}\otimes |\chi_2\rangle_{A_{1,2}A_{3,1}}\otimes |\chi_2\rangle_{A_{2,2}A_{3,2}}\\
	|\psi_2\rangle_{A_1A_2A_3} &= |\chi_3\rangle_{A_{1,1}A_{2,1}A_{3,1}}\otimes |\chi_3\rangle_{A_{1,2}A_{2,2}A_{3,2}}.
	\end{aligned}
\label{eq:weird-example}
\end{align}
In the state $\psi_1$ each party shares maximally entangled Bell states with the other two parties, whereas in $\psi_2$ the three parties share two copies of a tripartite GHZ state.
The spectra of all marginals of the two states coincide, and hence the same is true for all functions computed from these spectra, in particular entropies.
But the entanglement in $\psi_1$ and $\psi_2$ is rather different operationally, and this is witnessed mathematically e.g.~by computing the log-negativity $\log \|(\psi_k)_{A_iA_j}^{T_j}\|_1$ (with $X^{T_j}$ denoting the partial transpose over system $A_j$) of any two-party marginal of the two states.
We leave it as another open question for future research to determine functionals giving rise to persistence complexes that can detect the different entanglement structures of the states in \eqref{eq:weird-example}.

\subsection{Relative homology and strong subadditivity}\label{sec:reduced}
Finally, we discuss the concept of \textit{relative} homology, which is closely related to the reduced homology discussed in Section~\ref{sec:topsumm}, and particularly relevant for our setting of viewing a multipartite state as a simplicial complex.

To define this concept for a given complex $\Delta$, we fix a subcomplex $K \le \Delta$ and consider the relative chain groups $\cC_{i}(\Delta)/\cC_{i}(K)$, for which the induced boundary map $\partial_{k}'$ on the quotient group defines the relative homology $H_{k}(\Delta,K) = \ker \partial_{k}' / \text{im} \partial_{k+1}'$. 
We use superscripts to denote the subcomplex with respect to which the relative homology is taken.
For example, $\mathfrak{X}^{\Delta_{S}}$ is the integrated Euler characteristic (IEC) relative to the subcomplex $\Delta_{S} = 2^{S}$ for a given subset $S \subseteq \hal$. 
 
Quite remarkably, it turns out that the IEC of a relative persistence complex is intimately connected with strong subadditivity (SSA) \cite{lieb1973ssa}, a fundamental entropy inequality.
For a tripartite quantum state $\rho_{ABR}$, this inequality states that
\begin{align}
	I(A:B|R)\coloneqq S(AR) + S(BR) - S(R) - S(ABR)\geq 0,\label{eq:ssa}
\end{align}
where the quantity $I(A:B|R)$ is known as quantum conditional mutual information.

Let now $\cA=\lbrace A,B,R\rbrace$ be a tripartite system, and take the reduced homology on $\cA$ with respect to a subset $\lbrace A,B\rbrace$ for a given quantum state $\rho_{ABR}$ and the total correlation $C(\cdot)$ defined in \eqref{eq:total-correlation}.
Intuitively, the computation of the relative IEC $\mathfrak{X}^{\Delta_{\{A,B\}}}(\infty)$ in this setting amounts to omitting all contributions from simplices contained exclusively in $\lbrace A,B\rbrace$.
This leads to the following intriguing result:
\begin{theorem}
 	\label{thm:non-positivity-IEC}
 	Let $\rho_{ABR}$ be a tripartite quantum state and let $P^{\Delta_{A,B}}$ be the corresponding persistence module relative to the subcomplex $\Delta_{A,B}$, defined in terms of the total correlation $C(\cdot)$ from \eqref{eq:total-correlation}.
 	We then have that
 	\begin{align}
 		\begin{aligned}
 			\mathfrak{X}^{\Delta_{\{A,B\}}}(\infty)\coloneqq \lim_{q\to 1}\mathfrak{X}_{q}^{\Delta_{\{A,B\}}}(\infty) = -I(A:B|R)_{\rho}  \le 0.\label{eq:relative-iec}
 		\end{aligned}
 	\end{align} 
 \end{theorem}
 \begin{proof}
 	We denote once again by $\eps_{\mathrm{max}}$ the largest value of the filtration parameter at which the number of barcodes changes.
 	Functionally, the relative IEC includes contributions from all simplices except for those that occur only in $\Delta_{A,B}$. 
 	Written out, we have 
 	\begin{align}
 		\mathfrak{X}^{\Delta_{\{A,B\}}}(\eps_{\mathrm{max}}) &= \left(\eps_{\mathrm{max}} - C(R)\right)  -\left(\eps_{\mathrm{max}} - C(BR)\right)\notag\\
 		& \qquad {} -\left(\eps_{\mathrm{max}} - C(AR)\right)+ \left(\eps_{\mathrm{max}} - C(ABR)\right) \notag\\
 		&= C(AR) + C(BR) - C(R) - C(ABR)\notag\\
 		&= -S(AR) - S(BR) + S(R) + S(ABR)\notag\\
 		& = -I(A:B|R) \le 0,
 	\end{align}
 	which proves the equality in \eqref{eq:relative-iec}. 
 	The following inequality follows from \eqref{eq:ssa}.
 \end{proof}
 Theorem~\ref{thm:non-positivity-IEC} says that the relative IEC of a persistence complex defined on a tripartite quantum state is always non-positive because of SSA.
 This may suggest a deeper connection between entropy inequalities and the persistence complex of a multipartite quantum state, which will be explored in future work.
 
 We can apply the same idea to a bipartite state $\rho_{AB}$ and take the relative homology with respect to one of the vertices.
 This draws a connection to a special case of \eqref{eq:ssa} called subadditivity, $S(A)+S(B)-S(AB)\geq 0$.
 In this case, the relative IEC is non-negative:
 \begin{corollary}
 	\label{cor:bipartite-relative}
 	Let $\hal = \{A,B\}$. Then the relative persistent homology yields \begin{align}
 		\lim_{q\to 1}\mathfrak{X}_{q}^{\Delta_{\{A\}}}(\infty) = \lim_{q\to 1}\widetilde{\mathfrak{X}}_{q}(\infty) = I(A:B) \ge 0,
 	\end{align}  \par 
 \end{corollary}
 Note that Corollary~\ref{cor:bipartite-relative} cannot simply be obtained from Theorem~\ref{thm:non-positivity-IEC} via taking $R=1$ because of the different sets of simplices over which the IEC is summed.\par 

\section{Conclusion and future directions of research}\label{sec:conclusions}
In this work we have proposed a method of analyzing multipartite entanglement using persistent homology.
The key idea is to view the local systems of a multipartite quantum system $\cA = \lbrace A_1,\dots,A_n\rbrace$ as vertices on which we define a simplicial complex in terms of a correlation measure evaluated on subsets of vertices for a given quantum state: a simplex $J\subseteq \cA$ is added if the correlation measure does not exceed a given value $\eps$.
Varying this parameter $\eps$ then defines a filtration of simplicial complexes called the persistence complex, and homological tools can be used to track how the topology changes as a function of $\eps$.
The persistence complex is uniquely represented using persistence barcodes such as those in Figures \ref{fig:example}, \ref{fig:comparison-graph-states} and \ref{fig:comparison-osterloh-siewert}.
These barcodes encode the relevant topological information about the underlying persistence complex, and can be further processed to yield topological summaries such as the integrated Euler characteristic (IEC).
Our main results show that the persistence complex built with a deformed version of the total correlation yields a deformed version of the interaction information as its IEC.
Further specializing this to the total correlation based on the linear entropy yields the Minkowski length of the generalized Bloch vector as the IEC, which for an even number of qubits is known to be equal to the $n$-tangle, a well-known entanglement monotone.
Our persistent homology pipeline thus yields a topological summary that coincides with an operationally relevant measure in entanglement theory.
At the same time, the ``topological fingerprint'' of a quantum state encoded in its barcode conveys more information about the underlying entanglement structure than the topological summary alone. 
We illustrate this with examples of pairs of states with the same $n$-tangle, but lying in different SLOCC orbits, which is indicated by the corresponding barcodes (Figures \ref{fig:comparison-graph-states} and \ref{fig:comparison-osterloh-siewert}).
Finally, we draw a connection between the relative homology of a quantum state and the fundamental strong subadditivity property of the von Neumann entropy.

While our examples exhibited in Figures \ref{fig:comparison-graph-states} and \ref{fig:comparison-osterloh-siewert} demonstrate that barcodes may distinguish between different entanglement structures, it would be interesting to find a concrete \emph{operational} interpretation of these different topological features on the level of barcodes.
Moreover, the emergence of the $n$-tangle as a topological summary of a persistence complex suggests that our approach could lead to the identification of other entanglement monotones or operationally motivated correlation measures as topological invariants.
In principle, any multipartite entanglement functional satisfying monotonicity under partial trace can be used to build a persistence complex using our construction.
Examples of functionals with this property are the secrecy monotone \cite{Yang_Horodecki_Horodecki_Horodecki_Oppenheim_Song_2009}, the concentratable entanglement \cite{Beckey_Gigena_Coles_Cerezo_2021}, and what we term the ``powerset mean'' of any entanglement measure satisfying strong subadditivity (SSA). 

Another natural candidate for a functional monotonic under partial trace is a so-called generalized divergence \cite{leditzky2018approximate}. 
To see this, recall that in the special case $q\to 1$ the total correlation $C(J)_\rho$ of a state $\rho$ on $\cH_\cA$ evaluated on a subset $J\subseteq\cA$ can be written as
\begin{align}
	C(J)_\rho = D\Big(\rho_{J}\|\bigotimes\nolimits_{j\in J} \rho_{A_j}\Big) = \min_{\sigma_{A_i}} D\Big(\rho_{J}\|\bigotimes\nolimits_{j\in J} \sigma_{A_j}\Big),
	\label{eq:total-correlation-as-divergence}
\end{align}
where $D(\cdot\|\cdot)$ denotes the quantum relative entropy, and the minimization in the second equality is over all states $\sigma_{A_i}$ on the single systems $A_i$.
The expression \eqref{eq:total-correlation-as-divergence} is particularly useful, since the monotonicity property $C(I)\leq C(J)$ for $I\subseteq J$ is a direct consequence of the data processing inequality for $D(\cdot\|\cdot)$.
Moreover, it expresses the functional $C(J)$ as the relative entropy distance to the set of fully uncorrelated states.
This observation leads us to consider the following generalization.
Let $\bD(\cdot\|\cdot)$ be any \emph{generalized divergence} \cite{leditzky2018approximate}, that is, a functional on pairs of positive semidefinite operators satisfying the data processing inequality.
Then given a pair of quantum states $(\rho_\cA,\sigma_\cA)$, we may consider the persistent homology of $\rho_\cA$ relative to $\sigma_\cA$ induced by the functional $\bD(\cdot\|\cdot)$.
This approach is particularly amenable to the setting of a general resource theory consisting of a set of ``free states'' $\cF$ and a set of allowed operations~\cite{chitambar2019quantum}.
For example, in (bipartite) entanglement theory $\cF$ is the set of separable states, and LOCC is the set of allowed operations.
One then considers an entanglement measure called the \emph{relative entropy of entanglement} \cite{vedral1997quantifying}, defined as $E_R(\rho) \coloneqq \min_{\sigma\in\cF} D(\rho\|\sigma)$.
For a given resource theory with a fixed set $\cF$ of free states that is stable under taking partial traces and a suitable generalized divergence $\bD(\cdot\|\cdot)$, we may now consider the functional $\bE(\rho)\coloneqq \min_{\sigma\in\cF}\bD(\rho\|\sigma)$ for our persistent homology.
Inspired by Theorem~\ref{thm:n-tangle}, we expect our framework to yield interesting monotones for general resource theories, a thorough investigation of which will be the subject of future work. 

% \blue{Finally, there is a broad body of work that touches upon the categorical and topological aspects of mutual information (cite Mainiero, Leinster, vigneaux, bennequin, Bradley). In particular, the von Neumann entropy and mutual information appear as coboundaries of a cohomology defined in terms of ``informational structures" quite similar to the simplicial complexes defined here (cite). The integrated Euler characteristic is a simplistic yet powerful means of categorifying the interaction information, and we expect future work to clarify the connection with more formal category theoretic constructions. }

\paragraph*{Acknowledgments.}
The authors acknowledge helpful conversations with Henry Adams, Yuliy Baryshnikov, Jacob Beckey, Eric Chitambar, Jens Siewert, Juan Pablo Vigneaux, Michael Walter and Shmuel Weinberger. 
This research was partially supported through the IBM-Illinois Discovery Accelerator Institute.

\printbibliography[heading=bibintoc]

\end{document}